\documentclass[a4paper]{article}
\pdfoutput=1
 
\usepackage[margin=1.25in]{geometry}
\usepackage{graphicx,amssymb,amsmath}
\usepackage{verbatim}
\usepackage{enumerate}
\usepackage{color}
\usepackage{graphicx}
\usepackage{authblk}

\newcommand{\remove}[1]{{}}

\newtheorem{theorem}{Theorem}

\newtheorem{lemma}{Lemma}
\newenvironment{proof}{\paragraph{Proof:}}{\hfill$\square$}

\title{On Compatible Triangulations with a\\ Minimum Number of Steiner Points\thanks{Work is supported by the Natural Sciences and Engineering Research Council of Canada (NSERC).}
} 

\author{Anna Lubiw}
\author{Debajyoti Mondal}

\affil{Cheriton School of Computer Science, University of Waterloo, Canada 
    \texttt{\{alubiw,dmondal\}@uwaterloo.ca}}

\index{Lubiw, Anna}
\index{Mondal, Debajyoti}

\usepackage{amsmath}

\begin{document}

\maketitle

\begin{abstract}
Two vertex-labelled polygons are \emph{compatible} if they have the same clockwise cyclic ordering of vertices. The definition extends to polygonal regions (polygons with holes) and to triangulations---for every face, the clockwise cyclic order of vertices on the boundary must be the same. It is known that every pair of compatible $n$-vertex polygonal regions can be extended to compatible triangulations by adding $O(n^2)$ Steiner points. Furthermore, $\Omega(n^2)$ Steiner points are sometimes necessary, even for a pair of polygons. Compatible triangulations provide piecewise linear homeomorphisms and are also a crucial first step in morphing planar graph drawings, aka ``2D shape animation.'' An intriguing open question, first posed by Aronov, Seidel, and Souvaine in 1993, is to decide if two compatible polygons have compatible triangulations with at most $k$ Steiner points.  In this paper we prove the problem to be NP-hard for polygons with holes.  The question remains open for simple polygons.  
\end{abstract}

\section{Introduction}

For many computational geometry problems involving a polygon or polygonal region, the standard first step is to triangulate the region.  
However, for some problems, such as morphing of polygons, or finding a homeomorphism between polygons, the input consists of two polygons with a correspondence between them, and the desirable first step is to triangulate them in a consistent way. Unlike for a single polygon, it may be necessary to add new vertices, called ``Steiner points.'' 
Our paper is about this harder problem, which was called ``joint triangulation'' by Saalfeld~\cite{DBLP:conf/compgeom/Saalfeld87} and ``compatible triangulation'' by Aronov, Seidel, and Souvaine~\cite{DBLP:journals/comgeo/AronovSS93}.
%

Research on finding compatible triangulations 
 is motivated by applications in morphing~\cite{alamdari2016morph} and 2D shape animation~\cite{baxter2009compatible,surazhsky2004high},  
  and in computing piecewise linear homeomorphisms of polygons.

Throughout, we deal with vertex-labelled straight-line planar drawings. The most general input we consider is a polygon with holes (a polygonal region), where we allow a hole to degenerate to a single point.  
Two polygons are \emph{compatible} if they have the same clockwise cyclic ordering of vertices. Two polygonal regions $P_1$ and $P_2$ are \emph{compatible} if their outer polygons are compatible, and their holes are compatible, i.e.~each hole (considered as a polygon) in $P_1$ corresponds to a compatible hole in $P_2$. Note that the labelling provides the correspondence. 

 A \emph{triangulation $T(P)$ of a polygonal region $P$} is a subdivision of its interior region  into triangular faces.   The vertices of $T(P)\setminus P$ are called \emph{Steiner points}
 of $T(P)$.  A pair of triangulations $T(P_1)$  and $T(P_2)$ of compatible polygonal regions $P_1$ and $P_2$, respectively, are \emph{compatible}  if their faces are compatible, i.e.~every face of $T(P_1)$ (considered as a polygon) corresponds to a compatible face of $T(P_2)$.  Again, the labelling provides the correspondence. 
Figure~\ref{fig:intro} illustrates   a pair of compatible polygonal regions and their compatible triangulations. 

\begin{figure}[pt]
\centering
\includegraphics[width=.6\textwidth]{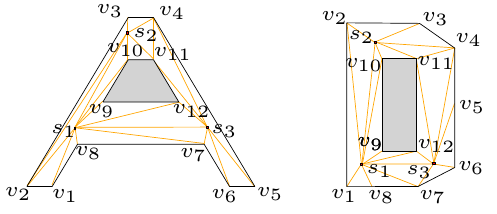}
\caption{Two compatible polygons, each with one hole (shaded gray), and compatible triangulations of them with 3 Steiner points.
}
\label{fig:intro}
\end{figure}

Two special cases of compatible triangulations were studied independently. Saalfeld in 1987~\cite{DBLP:conf/compgeom/Saalfeld87} considered the case of two rectangles each with $n$ points inside them (where the correspondence between the points is given) and showed that compatible triangulations always exist. Saalfeld's construction may require  an exponential number of Steiner points~\cite{sw1994}. Aronov et al., in 1993~\cite{DBLP:journals/comgeo/AronovSS93} considered the case of simple compatible polygons. They showed that 
there exist compatible triangulations with $O(n^2)$ Steiner points.
(A similar construction was given by Thomassen in 1983~\cite[Theorem 4.1]{thomassen1983deformations}.) Furthermore, Aronov et al.~gave an $O(n^2)$-time algorithm to compute such compatible triangulations, and they gave examples where $\Omega(n^2)$ Steiner points are necessary.
They posed as an open problem to decide if two polygons have a compatible triangulation with $k$ Steiner points, and observed that the case $k=0$ can be decided in polynomial time via dynamic programming.

\medskip
\noindent{\bf Our Result.} We show that it is NP-hard to decide if two compatible polygonal regions have compatible triangulations with at most $k$ Steiner points, where $k \in {\mathbb N}$ is given as part of the input. 

\medskip\noindent{\bf Further Background.}
There are a number of further results for the case of two simple polygons.  
Kranakis and Urrutia~\cite{DBLP:journals/ijcga/KranakisU99} 
gave an $O(n+r^2)$-time algorithm to find compatible triangulations of simple compatible polygons with $O(n+r^2)$ Steiner points, where $r$ is the number of reflex vertices.  
Gupta and Wenger~\cite{DBLP:journals/talg/GuptaW07} gave
a polynomial-time algorithm that provides an $O(\log n)$ approximation to the minimum number of Steiner points. 
A number of heuristic algorithms have been proposed---see e.g.,~\cite{baxter2009compatible,surazhsky2004high}.

There is also a line of research on the case of polygons with point holes (Saalfeld's problem). Souvaine and Wenger~\cite{sw1994} gave an $O(n^2)$-time algorithm to  compute compatible triangulations with $O(n^2)$ Steiner points, and asked if there is a polynomial-time algorithm to construct compatible triangulations with the minimum number of Steiner points.   
 Pach et al.~\cite{DBLP:journals/algorithmica/PachSS96}
 proved that $\Omega(n^2)$ Steiner points are sometimes necessary. 

For the case of general polygonal regions---which encompasses both the above special cases---Babikov et al.~\cite{DBLP:conf/cccg/BabikovSW97} gave an $O(n^2)$-time algorithm to compute compatible triangulations with $O(n^2)$ Steiner points.

One approach to computing compatible triangulations 
is to first compute a triangulation for one of the polygonal regions, and then draw its underlying graph into the other polygonal region using polylines for drawing edges. 
The edge bends give rise to the Steiner points. This idea relates to
the problem of drawing a planar graph on a given set of points, where the correspondence between vertices and the points is given. 
 Pach and Wenger~\cite{DBLP:journals/gc/PachW01} gave an $O(n^2)$-time algorithm to compute such an embedding with $O(n^2)$ bends in total, and this was extended to deal with a bounding polygonal region in~\cite{DBLP:journals/jgaa/ChanFGLMS15}. 


The version of the compatible triangulation problem where the correspondence between the two polygonal regions is \emph{not} given is also well-studied and very relevant in practice, e.g.~see~\cite{baxter2009compatible}.
In this setting, Aichholzer et al.~\cite{aichholzer2003towards} made the fascinating conjecture that for any two point sets each with $n$ points, of which $h$ lie on the convex hull, there is a mapping between them that permits compatible triangulations with no Steiner points. 



\section{Preliminaries}

Let $P$ be a polygon, possibly with holes.  
Two points $a,b$ in $P$ are \emph{visible} if the line segment between them lies  strictly inside $P$; they are \emph{1-bend visible} if there is a point $c$ inside $P$ that is visible to both $a$ and $b$.
%
 
 A \emph{dent} on the boundary of $P$ consists of three consecutive vertices $u,d,v$ of $P$
 such that $d$ is convex and $u,v$ are reflex vertices, e.g., see the polygon $P_1$ 
 in Figure~\ref{fig:dent}. We refer to $d$ as the \emph{peak}
 of the dent. The \emph{visibility region} of $d$
 consists of all the points inside $P$ that are 
 visible to $P$. 
 An \emph{inward dent} on the boundary of $P$ consists of three consecutive vertices $u,d,v$ of $P$
 such that $d$ is reflex and $u,v$ are convex vertices. 
The following simple lemma about dents in compatible triangulations of polygons
will be a key ingredient of our NP-hardness proof.

\begin{figure}[pt]
\centering
\includegraphics[width=.5\textwidth]{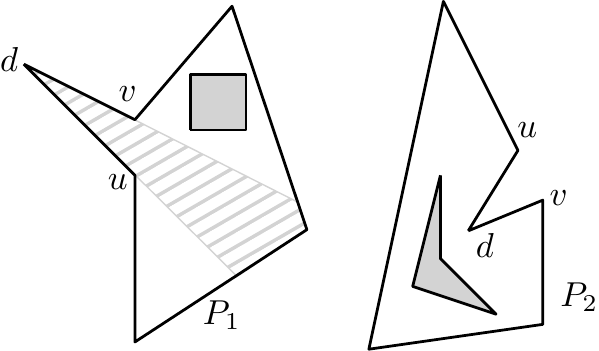}
\caption{Illustration for Lemma~\ref{lem:basic}. The visibility region of $d$
 is shown in gray stripes.}
\label{fig:dent}
\end{figure}

\begin{lemma}\label{lem:basic}
Let $P_1$ and $P_2$ be a pair of compatible polygons.
Assume that $P_1$ contains a dent $u,d,v$, and let 
 $\Psi$ be the visibility region of $d$ in $P_1$. 
 If $u,v$ are not 
 visible in $P_2$, then 
in any compatible triangulations
 $d$ must be adjacent either to a Steiner point  or   
a vertex (except $u$ and $v$) 
 inside $\Psi$. 
\end{lemma}
\begin{proof}
Any triangulation of $P_1$ (even with Steiner points) must use the edge $(u,v)$ or an edge incident to $d$.  In compatible triangulations of $P_1$ and $P_2$ the edge $(u,v)$ is ruled out, and therefore $d$ must be adjacent to a Steiner point or a vertex in $\Psi\setminus \{u,v\}$.
\end{proof}


\section{NP-Hardness}
\label{sec:hardness}

\begin{figure*}[pt]
\centering
\includegraphics[width=\textwidth]{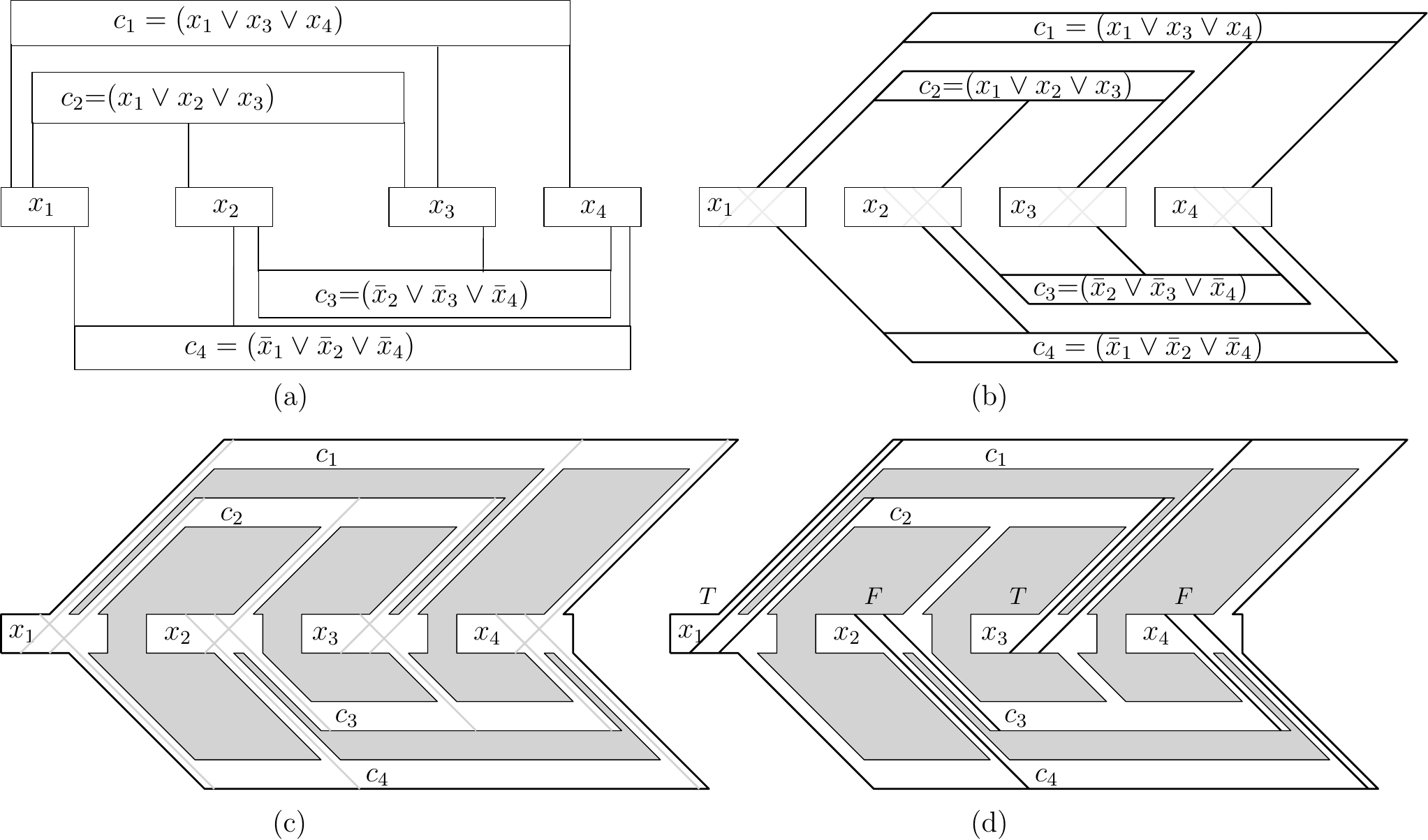}
\caption{(a) An instance $I$ of MRP-3SAT, and the corresponding drawing $\Gamma$. (b) $\Gamma'$.
  (c)--(d) Illustration for the hardness reduction.}
\label{np}
\end{figure*}

In this section we prove that given a pair of compatible polygonal regions $P_1,P_2$, 
and $k \in \mathbb N$,
it is NP-hard to decide if there are compatible triangulations of $P_1$ and $P_2$ with at most $k$ Steiner points.

We reduce from the monotone rectilinear planar $3$-SAT problem (MRP-3SAT), which is NP-complete~\cite{deBerg2010}.
 The input of an MRP-3SAT 
 instance $I$ is a collection $C$ of clauses over a set $U$ of Boolean variables 
 such that each clause contains at most three literals, and 
is either \emph{positive} (consists of only positive literals),
 or \emph{negative} (consists of only negative literals). 
 Moreover, the corresponding \emph{SAT-graph} $G_I$ (the bipartite graph with
 vertex set $C \cup U$ and edge set $\{(c,x) \in C \times U :$ $x$ appears in $c \}$) admits a planar
 drawing $\Gamma$ satisfying the following properties:  
\begin{enumerate}
\item[] {\bf -} Each vertex in $G_I$ is drawn as an axis-aligned rectangle in $\Gamma$. \newline
{\bf -} All the rectangles representing variables lie along a horizontal line $\ell$. 
\newline
{\bf -} The rectangles representing positive (respectively, negative) clauses lie above (respectively, below) $\ell$.  \newline
{\bf -} Each edge $(c,x)$ of $G_I$ is drawn as a vertical  line segment that connects the rectangles corresponding to $c$ and $x$, e.g., see Figure~\ref{np}(a).  \newline
\end{enumerate} 
The MRP-3SAT problem asks 
whether there is a 
truth assignment for $U$
 satisfying all clauses in $C$.

Given an instance $I=(U,C)$ of MRP-3SAT, we construct a pair of compatible polygonal regions $P_1$ and $P_2$ 
 such that they admit  compatible triangulations with at most $5|C|$ Steiner points, if and only if
 $I$ is satisfiable.

\paragraph{Idea of the reduction:} We first ensure that every clause 
 in $I$ has exactly three literals, by duplicating literals if necessary.
 Let the resulting instance be $I'$.
 It is straightforward to observe that $I'$ is also an instance of MRP-3SAT, and 
  $I'$ is satisfiable if and only if $I$ is satisfiable.
  Let $\Gamma$ be the drawing corresponding to $G_{I'}$. 
 
We   modify the drawing $\Gamma$ such that 
 the edges and   vertices corresponding to the positive (resp., negative) 
 clauses become parallelograms, slanted $45^\circ$ (resp., $-45^\circ$) to the right,
 e.g., see Figure~\ref{np}(b). 
For each clause $c \in C$, let $R(c)$ denote the parallelogram corresponding to $c$.  We call $R(c)$ the ``clause region''.
 For each variable $u\in U$, let $B(u)$ 
 denote the rectangle corresponding to $u$.
We call $B(u)$ the ``variable region''.
We call the edges of $G_{I'}$ \emph{connectors} and we call
the connectors
 that are incident to the top  (resp., bottom) side of $B(u)$ \emph{top} (resp., \emph{bottom})
 \emph{connectors} of $B(u)$.
We ensure that the extension of every top connector  intersects the
 extensions of all the bottom connectors inside $B(u)$. Let the resulting drawing be $\Gamma'$.
 We construct $P_1$ and $P_2$ by modifying two distinct copies of $\Gamma'$.

We prove that in any compatible triangulations with $5|C|$ Steiner points,
for each clause $c$, there is a triangulation edge $e_c$ that lies along one of the  connectors incident to the clause region.  
%
%
 If $c$ is positive (resp., negative)
 then we can set 
 the variable corresponding to $e_c$ to true (resp., false)
and this will satisfy the clause.   We get a valid truth-value assignment because a variable region cannot contain extensions of both top and bottom connectors.   
 Figures~\ref{np}(c)--(d) illustrate a satisfying truth assignment for $I$. 
 On the other hand, given a 
 satisfying truth assignment, we show how to find  compatible 
  triangulations for $P_1$ and $P_2$ using $5|C|$ Steiner points.

\begin{figure}[pt]
\centering
\includegraphics[width=.6\textwidth]{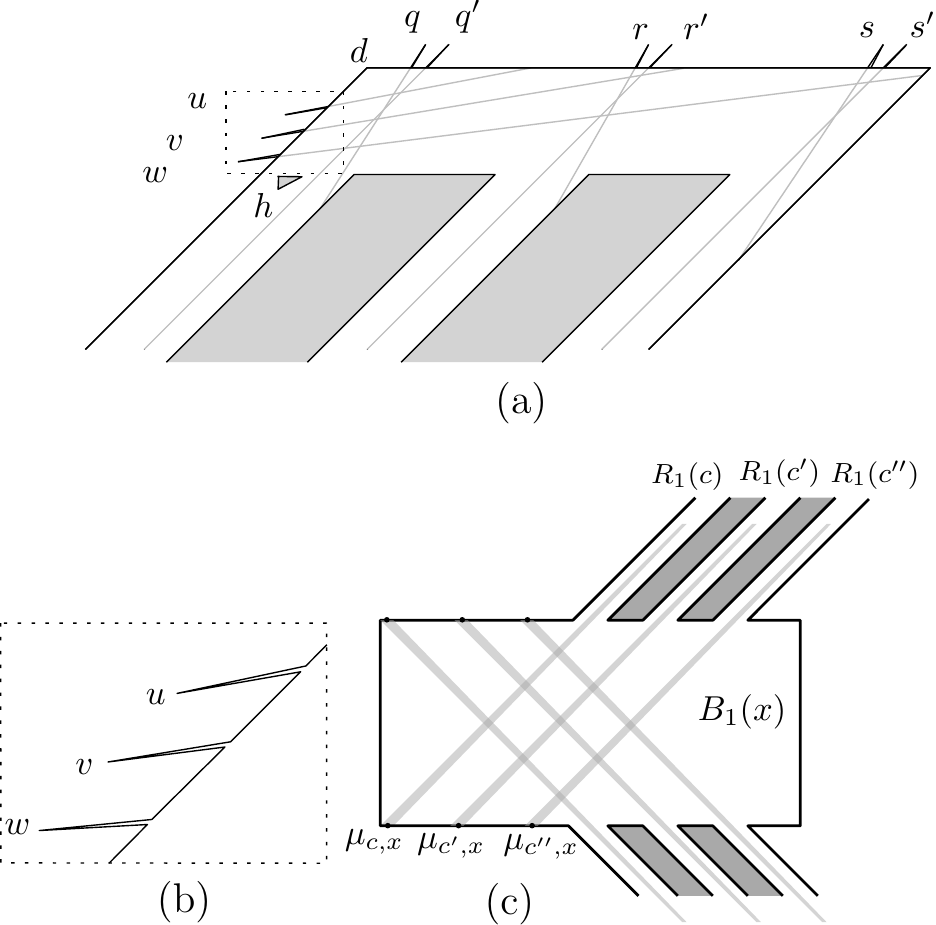}
\caption{(a) A clause region in $P_1$. (b) A close-up
 of the dents corresponding to $u,v,w$.  (c)   Illustration for $B_1(x)$. }
\label{cg}
\end{figure}
   
\subsection{Construction of Polygonal Region $P_1$}
\label{P1}
We modify a copy $\Gamma'_1$ of $\Gamma'$ to construct $P_1$.
First we create a \emph{channel} of small non-zero width around each connector so that we have a polygon with holes.  We denote the copies of $R(c)$ and $B(u)$ in $P_1$ by $R_1(c)$ and $B_1(u)$.
 We create nine dents with peaks $u,v,w,q,q',r,r',s,s'$ in the boundary of $R_1(c)$,
 as shown in Figures~\ref{cg}(a)--(b). The visibility region of each dent is illustrated using
  gray straight lines. 

As illustrated in  Figure~\ref{cg}(a), 
we  place a hole $h$ in the leftmost channel of $R_1(c)$, 
not intersecting the visibility regions of the peaks $u,v,w,q,q',r,r',s,s'$.
 We refer the reader to Section~\ref{appA} for the formal details of the construction 
and the precise placement of $h$.

We now modify the rectangles that correspond to the variables. 
 Let $x$ be a literal and let $B_1(x)$ be the corresponding rectangle 
 in $\Gamma'_1$. See Figure~\ref{cg}(c). For every positive   (resp., negative) 
 clause $c$ containing $x$, one or more\footnote{Recall that $c$ may contain duplicates of a literal.}
 visibility regions corresponding to the peaks of $R_1(c)$ enter $B_1(x)$.
We ensure that the visibility regions entering from the top (resp., bottom) of $B_1(x)$ are disjoint and only intersect the bottom (resp., top) side of $B_1(x)$.  For each clause $c$ containing $x$, we construct a vertex $\mu_{c,x}$ on the side of $B_1(x)$ such that $\mu_{c,x}$ is visible to the corresponding peak of $R_1(c)$.  
We refer to these newly constructed points as the \emph{$\mu$-points} of $B_1(x)$.


\begin{figure*}[htb]
\centering
\includegraphics[width=\textwidth]{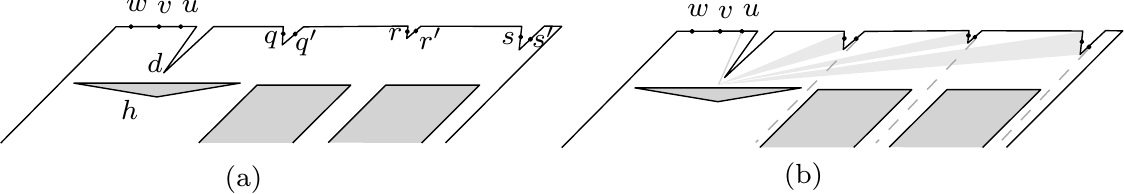}
\caption{(a) A clause region in $P_2$. (b) The vertex $u$ is not 1-bend 
 visible to $q',r',s'$.  }
\label{cg2}
\end{figure*}

\subsection{Construction of Polygonal Region $P_2$}
\label{P2}
We modify a copy $\Gamma'_2$ of $\Gamma'$ to construct $P_2$.
As in the construction of $P_1$, we create a channel of small non-zero width around each connector so that we have a polygon with holes.  We denote the copies of $R(c)$ and $B(u)$ in $P_2$ by $R_2(c)$ and $B_2(u)$.
 We create four inward dents on the boundary of $R_2(c)$, and place
 the points $u,v,w,d,q,q',r,r',s,s'$, as shown in Figure~\ref{cg2}(a).
 Finally, we place the hole $h$
 ensuring that no peak in $\{u,v,w\}$ is 1-bend visible to 
 $\{q',r',s'\}$, e.g., see  Figure~\ref{cg2}(b).  
 We refer the reader to Section~\ref{appA} for the formal details of the construction.
 
We now modify the rectangles that correspond to the literals. 
 Let $x$ be a literal and let $B_2(x)$ be the corresponding rectangle 
 in $\Gamma'_2$. The modification for $B_2(x)$
 is analogous to that of $B_1(x)$. Specifically, for every
 visibility region (of some peak $p\in \{q',r',s'\}$) that intersects 
 the box $B_1(x)$ in $\Gamma'_1$, we construct a point $\mu$
 on the boundary of box $B_2(x)$ such that  $\mu$ and $p$
 are  visible in  $P_2$. Figure~\ref{cg2}(b)
 illustrates such   visibilities with dashed lines. 

\subsection{Properties of Compatible Drawings}
In this section we prove some key 
properties of compatible triangulations $T(P_1)$ and $T(P_2)$ of $P_1$ and $P_2$, respectively.
For clause $c$, let $\overline R_1(c)$ be the clause region $R_1(c)$ plus its three attached channels.

\begin{lemma}
If $c$ is a clause such that no peak $q',r',s'$ is adjacent in $T(P_1)$ to a point 
outside $\overline R_1(c)$,
then there are at least 6 Steiner points in $\overline R_1(c)$.
\label{lem:positive0}
\end{lemma}
\begin{proof}
Consider the 9 points $\{u,v,w,q,q',r,r',s,s'\}$.  In $P_1$ each point in this set is the peak of a dent, so
by Lemma~\ref{lem:basic}, each of these 9 points must be adjacent in $T(P_1)$ to a vertex or a Steiner point. 
The only vertices visible to any of the 9 peaks are the $\mu$-points visible to $q',r',s'$, but they lie outside  $\overline R_1(c)$.  We assumed there is no edge from $q',r',s'$ to a point outside $\overline R_1(c)$.  The other 6 peaks are not visible to any point outside $\overline R_1(c)$. 
Thus each of the 9 peaks must be adjacent to a Steiner point in $\overline R_1(c)$. 
No point in $\overline R_1(c)$  is visible to more than two peaks.  Thus we need at least 
$\lceil \frac{9}{2} \rceil = 5$ Steiner points.  
The only way that 5 Steiner points suffice is to use 4 Steiner points that are each adjacent to two peaks.   Pairs of peaks that are visible to a common point in both $P_1$ and $P_2$ are indicated by edges in the graph $H$ shown in Figure~\ref{fig:cover}(a).  We require a matching of size 4 in $H$.  Observe that $H$ is bipartite so the maximum size of a matching is equal to the minimum size of a vertex cover.  The set $\{q,r,s\}$ is a vertex cover of size 3.  Thus there is no matching of size 4, and the Lemma follows.
\end{proof}

\begin{lemma}
For any clause $c$, there are at least 5 Steiner points in $\overline R_1(c)$.
\label{lem:positive}
\end{lemma}
\begin{proof}  Consider the triangulation of $P_1$.
The case where no peak $q',r',s'$ has an incident edge to a point outside 
$\overline R_1(c)$ is covered by Lemma~\ref{lem:positive0}.  It remains to consider the cases when there is such an edge. 

Our argument will be partly about the graph $H$ (in Figure~\ref{fig:cover}(a)) of pairs of peaks that are visible to a common point in both $P_1$ and $P_2$, and partly about the geometry of $P_1$.
First we note that the argument used above in the proof of Lemma~\ref{lem:positive0} can be strengthened to show that if we use just one edge from a peak to a point outside the clause region then we still need 5 Steiner points inside the region.  In graph $H$, observe that if one of $q',r',s'$ is removed, then we have 8 vertices, and a maximum matching of size 3, which means that we can use 3 edges (Steiner points) to cover 6 vertices, leaving 2 vertices that need one Steiner point each, for a total of 5 Steiner points.   It remains to consider the cases where at least two of the points $q',r',s'$ have an incident edge to a point outside the clause region.  We deal with the case where $q'$ has such an edge and the case where $r'$ has such an edge but $q'$ does not.   

Suppose there is an edge $e$ from $q'$ to a point outside $\overline R_1(c)$.  Observe that edge $e$ cuts off the visibility regions of $v$ and $w$.  
The effect on graph $H$ is to remove the edges of $H$ incident to $v$ and $w$, e.g., see Figure~\ref{fig:cover}(b).  Thus we need one Steiner point for each of $v$ and $w$, one Steiner point for $r$ (irrespective of how $r'$ is connected), one Steiner point for $s$ and one more for $u$, a total of at least 5.

Next suppose there is no edge from $q'$ to a point outside of $\overline R_1(c)$, but there is an edge $e'$ from $r'$ to a point outside of $\overline R_1(c)$.   The edge $e'$ cuts off the visibility region of $w$.   The effect on graph $H$ is to remove the edges $(w,r)$ and $(w,s)$, e.g., see Figure~\ref{fig:cover}(c).    We then need a Steiner point for $s$ (irrespective of how $s'$ is connected), and for the remaining 6 vertices $\{u,v,w,q,q',r\}$, we have a subgraph with a minimum vertex cover $\{q,r\}$ of size 2, thus a maximum matching of 2 edges (Steiner points) to cover 4 vertices, leaving 2 vertices that need one Steiner point each, for a total of 5 Steiner points.
\end{proof}

\begin{lemma}
If $T(P_1)$ and $T(P_2)$ use $5|C|$ Steiner points each, then for any clause $c$, there is an edge in $T(P_1)$ from at least one of $q', r', s'$ to a $\mu$-point.
%
\label{lem:positive2}
\end{lemma}
\begin{proof}
By Lemma~\ref{lem:positive} every  region $\overline R_1(c)$ has at least 5 Steiner points.  Thus every such region must have exactly 5 Steiner points and there are no Steiner points 
in the variable regions. 
Suppose there is a clause $c$ such that $T(P_1)$ has no edge from $q', r'$ or $s'$ to a $\mu$-point.  Then there is no edge from  $q', r'$ or $s'$ to a point outside $\overline R_1(c)$.  But then by Lemma~\ref{lem:positive0} the clause region must have at least 6 Steiner points, a contradiction.
\end{proof}


\begin{figure}[ptb]
\centering
\includegraphics[width=.5\textwidth]{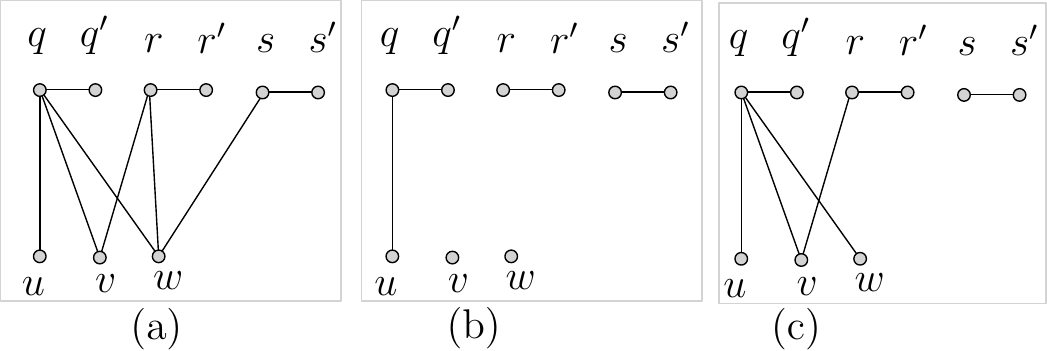}
\caption{(a) Graph $H$ of pairs of peaks that that are visible to a common Steiner point in both $P_1$ and $P_2$. (b)--(c) Illustration for Lemma~\ref{lem:positive}.
}
\label{fig:cover}
\end{figure}

\subsection{Reduction}
\begin{theorem}
\label{thm:nph}
The following problem is NP-hard:
Given a pair of compatible polygonal regions $P_1,P_2$, 
and $k \in {\mathbb N}$, decide if $P_1$ and $P_2$ have compatible triangulations with at most $k$ Steiner points.
\end{theorem}
\begin{proof}
Let $I=(U,C)$  be an instance of MRP-3SAT, and let $P_1$ and $P_2$ be the corresponding 
 compatible polygons, as described in Sections~\ref{P1}--\ref{P2}.
  Section~\ref{appA} presents further details on how to construct $P_1$ and $P_2$ using a polynomial number of bits, so this is a polynomial-time reduction.
 We now  prove that $P_1$  and $P_2$ admit a pair of  compatible triangulations, each 
 with at most $5|C|$ Steiner points, if and only if $I$ admits a satisfying truth assignment.

\begin{figure}[pt]
\centering
\includegraphics[width=.5\textwidth]{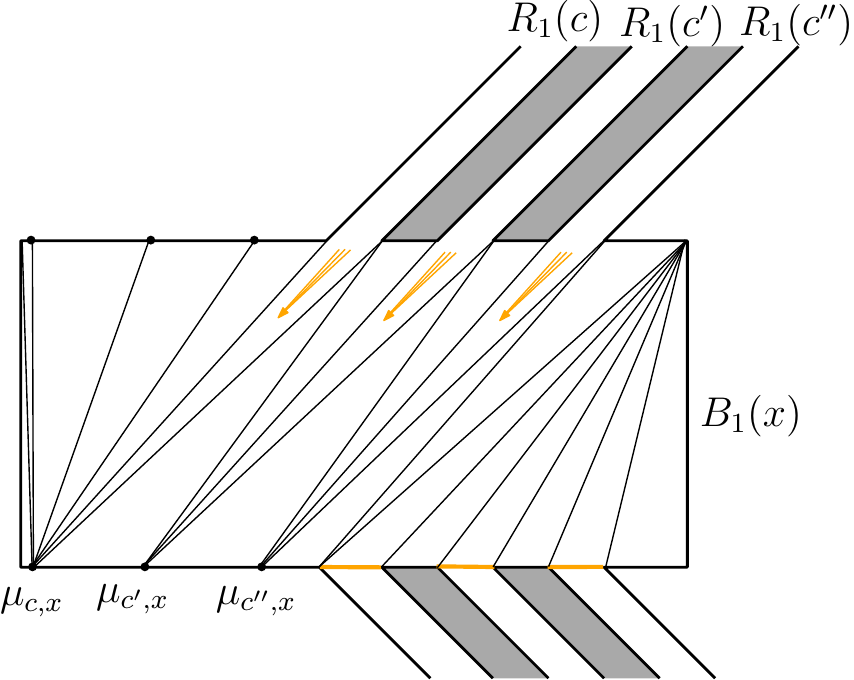}
\caption{A triangulation for $B_1(x)$, where $x={\it true}$.}
\label{fig:easy-dir3}
\end{figure}

We first assume that $P_1$ and $P_2$ admit compatible triangulations with at most $5|C|$ Steiner points.   By Lemma~\ref{lem:positive2}, for any clause $c$ there is an edge in the triangulation of $P_1$ from at least one peak $z \in \{q',r',s'\}$ to  a $\mu$-point, say $\mu_{c,x}$. 
We use the edge  $(z,\mu_{c,x})$   to 
assign a truth value to variable $x$.
 If $c$ is a positive (resp., negative) clause, then
 we set $x$ to true (resp., false).
Clearly we have satisfied each clause.  
 If there is a variable $u$ whose truth value is not assigned yet,
 then setting the truth value of $u$ arbitrarily would still keep 
 the clauses satisfied. 
 It remains to show 
that the truth-value assignment is consistent.
Suppose
 there is a variable  $u$ such that some clause $c$ forces $u$ to be true,
 and some other clause $c'$ forces $u$ to be false.
 Without loss of generality we may assume that $c$ is positive and $c'$ is
 negative. Consequently, in each of $R_1(c)$ and $R_1(c')$, there exists a peak 
 that is incident to some $\mu$-point in $B_1(x)$. By construction of the 
 $\mu$-points in $B_1(x)$, the two corresponding edges cross, a contradiction.

Assume now that $I$ admits a satisfying truth assignment.  We will 
 find corresponding compatible triangulations of $P_1$ and $P_2$. For each variable $x$,
 if $x$ is set to true, then we close the channels of the 
 negative clauses and construct the compatible triangulations of  
 the rectangles $B_1(x)$ and $B_2(x)$ using
 the $\mu$-points on the bottom side of 
 these rectangles, e.g., see  Figure~\ref{fig:easy-dir3}.
 The construction when  $x$ is set to false is symmetric.

Since every clause contains at least one true literal, for every 
 clause $c$, there exist one or more peaks in $P_1$ that are 
 visible to their corresponding $\mu$-points. 
 We show that in each scenario, the corresponding clause gadgets can be triangulated in a compatible fashion.

 Section~\ref{appA} contains 
 the remaining details of the proof of Theorem~\ref{thm:nph}. These details concern the second half of the proof, where we assume that $I$ admits a satisfying truth assignment and we find corresponding compatible triangulations of $P_1$ and $P_2$. As explained in the main text,  for each variable $x$, if $x$ is set to true, then we close the channels of the 
 negative clauses and construct the compatible triangulations of  
 the rectangles $B_1(x)$ and $B_2(x)$ using the $\mu$-points on the bottom side of these rectangles, e.g., see  Figure~\ref{fig:easy-dir3}. The construction when  $x$ is set to false is symmetric.

Since every clause contains at least one true literal, for every 
 clause $c$, there exist one or more peaks in $P_1$ that are 
 visible to their corresponding $\mu$-points. 
It remains to show that in each scenario, the 
corresponding clause gadgets can be 
 triangulated in a compatible fashion.  
%
%
%
%
 We only describe the case when $c$ is a positive clause. 
 The case when $c$ is negative is symmetric.
 
Let $x_{q},x_{r},x_{s}$ be the literals of $c$, and 
  assume that their corresponding peaks $q',r',s'$ appear in this  
 order  from left to  right in $R_1(c)$. 

\textbf{Case 1 ($x_{q}={\it true}$):} 
  Figures~\ref{fig:easy-dir}(a)--(b) illustrate the 
  compatible triangulations of $R_1(c)$ and $R_2(c)$ for the 
  case when  $x_{q}={\it true}$ and $x_{r}=x_{s} = {\it false}$. 
  Figures~\ref{fig:easy-dir}(c)--(d) illustrate the 
  compatible triangulations of $R_1(c)$ and $R_2(c)$ for the 
  case when  $x_{q}={\it true}$ and $x_{r}=x_{s} = {\it true}$.
  The scenarios when $x_{r}={\it true}$ and $x_{s}={\it false}$, or vice versa,
  can be handled  by switching between the local configurations 
  corresponding to the true and false values. 

\textbf{Case 2 ($x_{q}={\it false}$, $x_{r}={\it true}$):} 
  Figures~\ref{fig:easy-dir2}(a)--(b) illustrate the 
  compatible triangulations of $R_1(c)$ and $R_2(c)$ for the 
  case when  $x_{s} = {\it true}$.   The scenario when $x_{s} = {\it false}$ 
  can  be handled  by switching between the local configurations 
  corresponding to the true and false values. 
  
\textbf{Case 3 ($x_{q}=x_{r}={\it false}$, $x_{s}={\it true}$):}   
  The  compatible triangulations of $R_1(c)$ and $R_2(c)$ for this case
   are illustrated in Figures~\ref{fig:easy-dir2}(c)--(d).

\begin{figure}[pt]
\centering
\includegraphics[width=.9\textwidth]{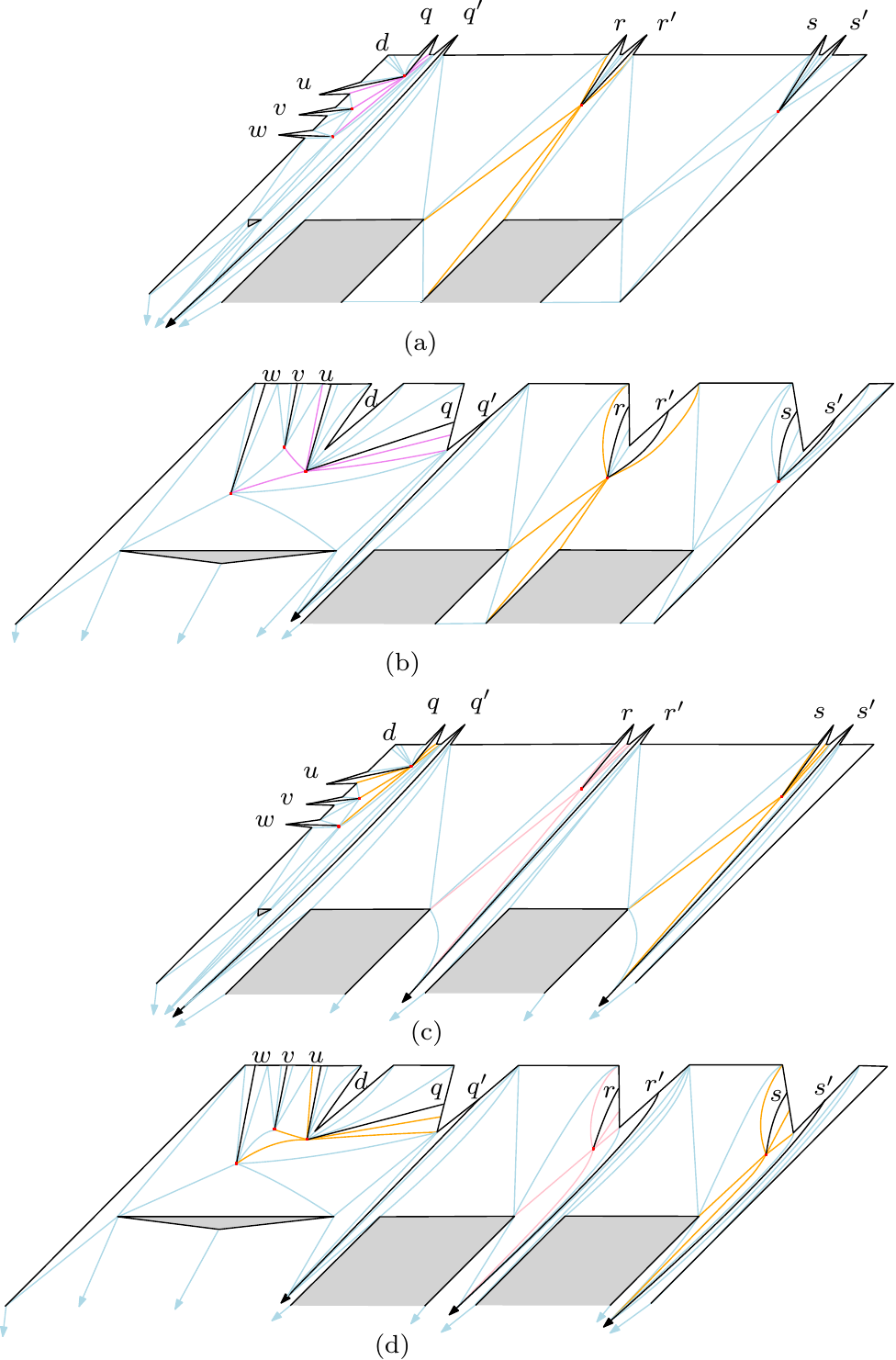}
\caption{Illustration for  Case 1. Colors are used  to 
 better illustrate the correspondence between the two drawings.  Some edges are drawn curved but that is only to make the drawing more readable.}
\label{fig:easy-dir}
\end{figure}

\begin{figure}[pt]
\centering
\includegraphics[width=.9\textwidth]{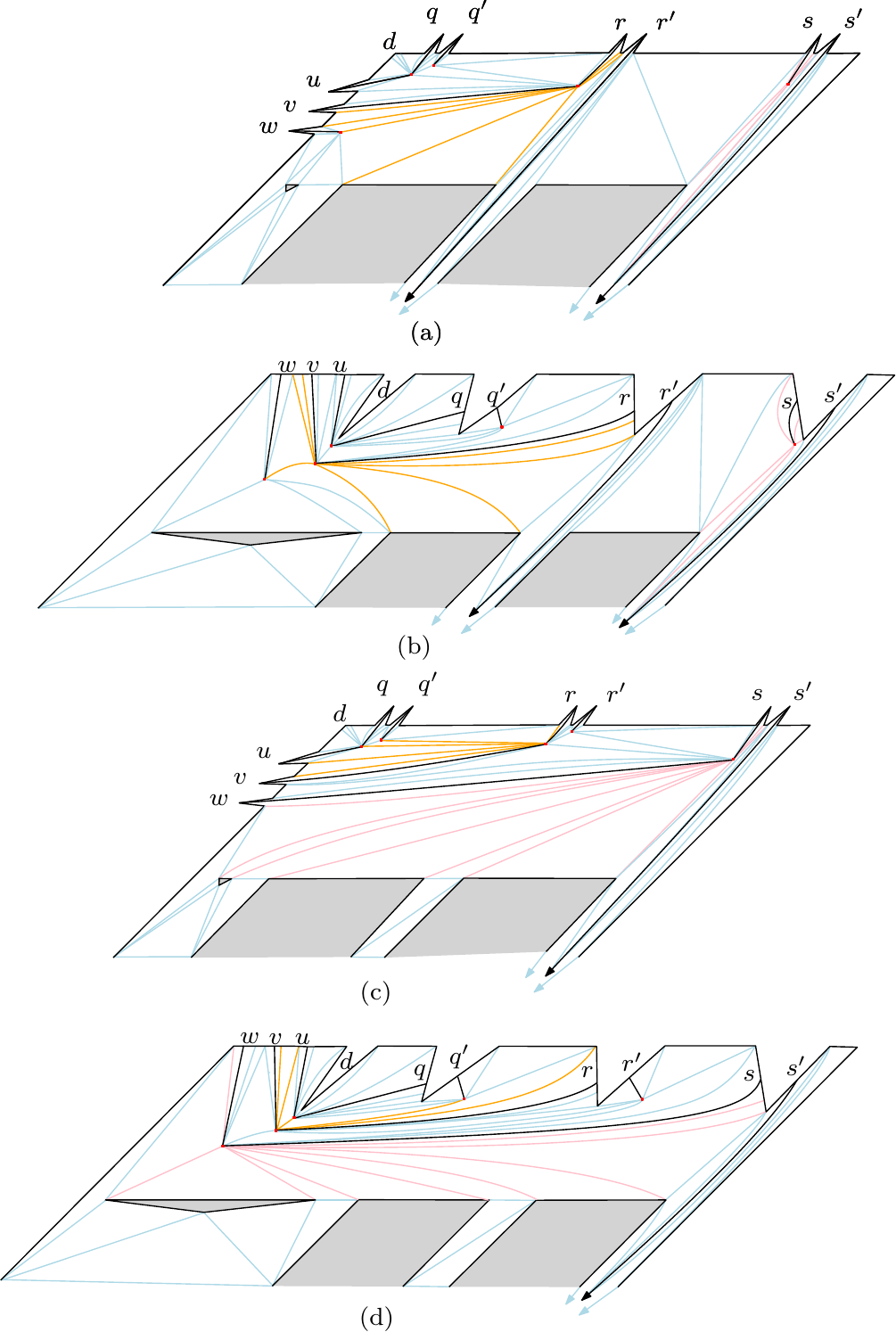}
\caption{Illustration for (a)--(b) Case 2, and  (c)--(d) Case 3. Some edges are drawn curved but that is only to make the drawing more readable.}
\label{fig:easy-dir2}
\end{figure}
%
%
\end{proof}


\section{Polynomial-time Construction}
\label{appA}
In this section we describe  the construction details of $P_1$ and $P_2$. Furthermore, we show that the construction can be accomplished in polynomial time, in particular, with a polynomial number of bits for the coordinates of all points.

\begin{figure}[pt]
\centering
\includegraphics[width=.7\textwidth]{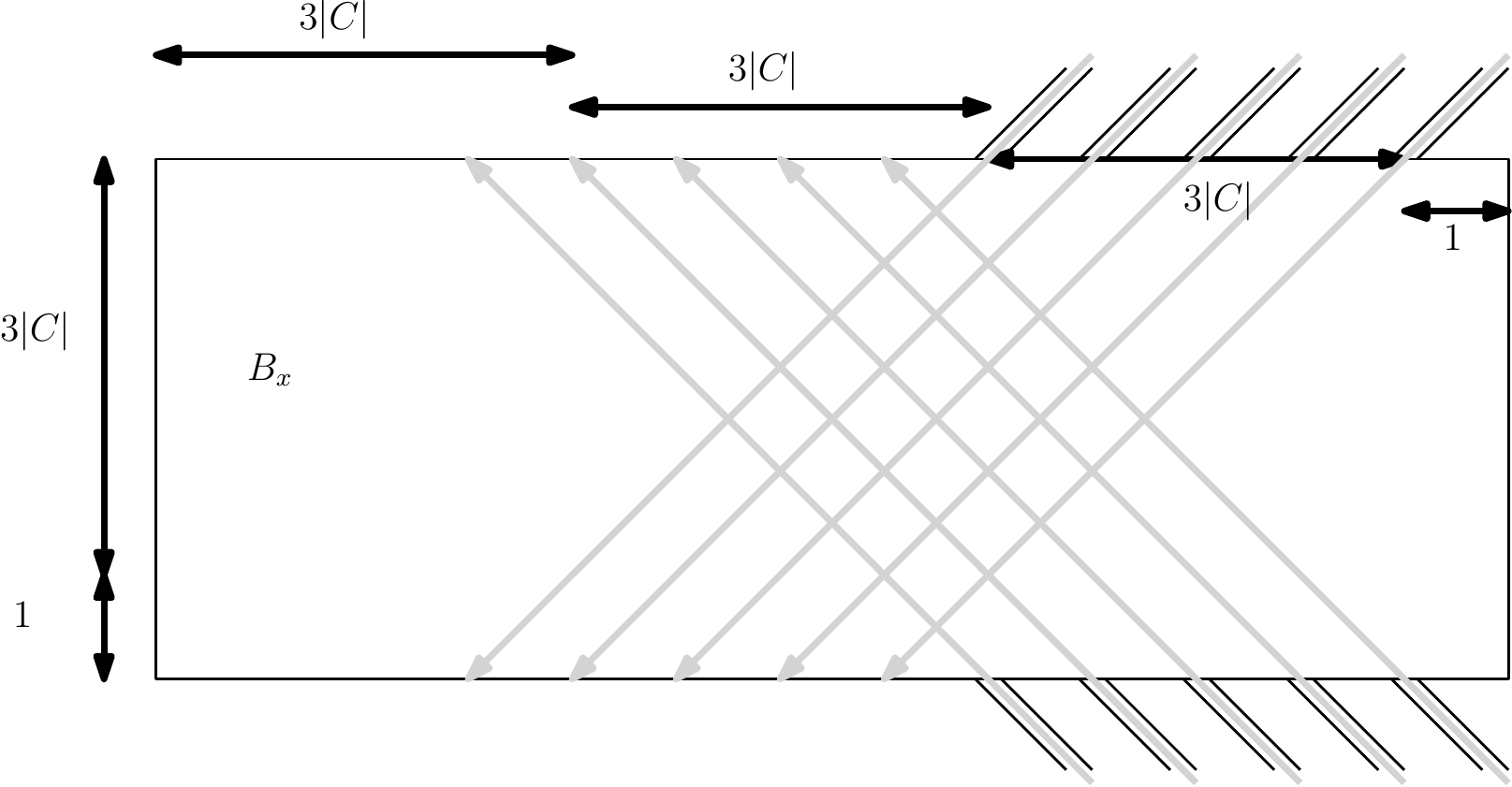}
\caption{Transforming $\Gamma$ into $\Gamma'$.}
\label{fig:easydetails}
\end{figure}

The construction has two stages.  In the first stage we construct $\Gamma'$ from 
$\Gamma$ (see Figures~\ref{np}(a)--(b)) and in the second stage we construct $P_1$ and $P_2$ from copies of $\Gamma'$ by adding the appropriate dents (see Figures~\ref{cg} and~\ref{cg2}).

For the first stage we claim that $\Gamma'$ can be constructed on a polynomial-sized grid (thus, with a logarithmic number of bits per coordinate).  
We make each variable rectangle $B(x)$ of width $9|C|+1$ and height $3|C|+1$,  as 
 illustrated in Figure~\ref{fig:easydetails}. 
The length allocated for channel attachments is $3|C|$ to incorporate possible variable duplicates.
 The distance between successive $\mu$-points on the bottom (top) of  $B(x)$ is at least 1.
 We make each clause parallelogram $R_c$ of height 1.
 The resulting drawing $\Gamma'$ has height $O(|U|+|C|)$ and width $O(|U||C|)$.
 
 We now turn to the second stage, the construction of $P_1$ and $P_2$.  
 We refer the reader to Figure~\ref{fig:details}, which illustrates for each positive clause $c$, the 
 correspondence between  the vertices of $R_1(c)$ and $R_2(c)$.  
Observe that the left, top, and right sides of holes $h_1$ and $h_2$ have no dents added to them and remain straight line segments. 

 We construct the peaks $u,v,w,q,q',r,r',s,s'$ of $R_1(c)$ iteratively.  
Our plan is to first construct the visibility lines of the peaks (see Figure~\ref{clause-gadget-constr}) and later enlarge these to visibility cones.  
Start by placing the points $q', r', s'$ above the top boundary of $R_1(c)$ 
  so that the lines from them to the corresponding $\mu$-points are parallel to the channel sides and centered in the channels.  
  Next, place points $u,v,w$  to the left of $R_1(c)$ and in the top half of $R_1(c)$, and choose points $u',v',w'$ on the boundary of $B(x)$ to be the endpoints 
 of the visibility lines emanating from $u,v,w$ respectively.  
Specifically, choose $u'$ and $v'$ by taking the midpoints of the tops of $h_1$ and $h_2$ and projecting upward at $45^\circ$.
%

\begin{figure}[ht]
\centering
\includegraphics[width=.8\textwidth]{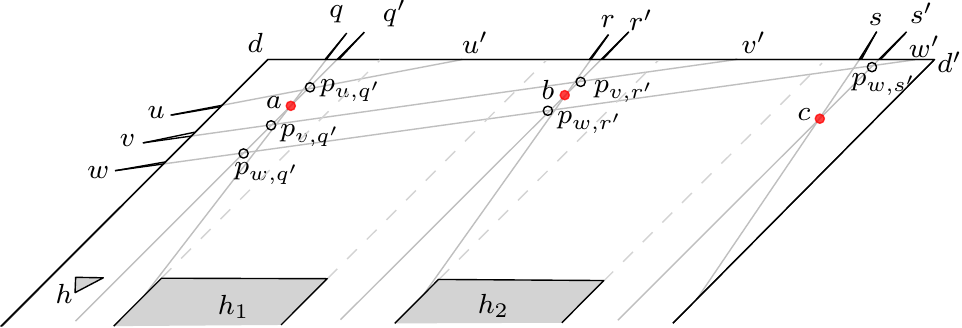}
\caption{Details of the construction of $R_1(c)$. 
}
\label{clause-gadget-constr}
\end{figure}

\begin{figure}[ht]
\centering
\includegraphics[width=.8\textwidth]{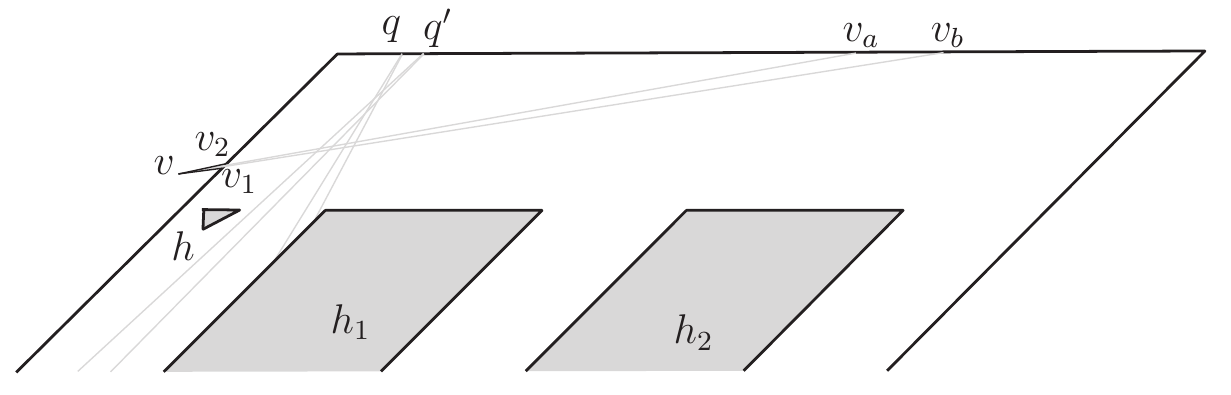}
\caption{Details of visibility cone construction for $v$.}
\label{vis-cone}
\end{figure}

We can explicitly compute the equations of the 6 visibility lines emanating from $q',r',s',u,v,w$ and 
we can compute the intersection points formed by them. 
 Let $p_{u,q'}$ be the intersection point of the visibility line of $u$ and the visibility line of $q'$, etc.
 We can next choose points $a,b,c$ on the visibility lines of $q',r',s'$, respectively, such that each point
 is in-between the appropriate $p$ points.  
For example, point $a$ is the midpoint of $p_{v,q'}$ and $p_{u,q'}$.
 From these, we can construct points $q,r,s$ and their visibility lines through $a,b,c$ to appropriate points on the channels.
Observe that because $a,b,c$ lie in the upper half of $B_1(x)$, points $q,r,s$ lie in the extensions of the channels, and remain to the right of $d,u',v'$, respectively.
 
It remains to enlarge the visibility lines to cones.  We can do this by explicitly
 computing a tolerance $\tau$ such that if the width of every cone is at most
 $\tau$ inside $P_1$ then no visibility cone will contain points it should not,
 and no two visibility cones will intersect when they should not.  Since we only 
 need a lower bound on $\tau$, this can be done with a polynomial number of bits.
 Finally, from $\tau$ we can explicitly choose the points where each visibility
 cone intersects the boundary of $P_1$ (see for example points $v_a$ and $v_b$
 in Figure~\ref{vis-cone}), and from these we can compute the dent vertices for
 each peak (see points $v_1$ and $v_2$ in the same figure).

 \begin{figure}[pt]
\centering
\includegraphics[width=.9\textwidth]{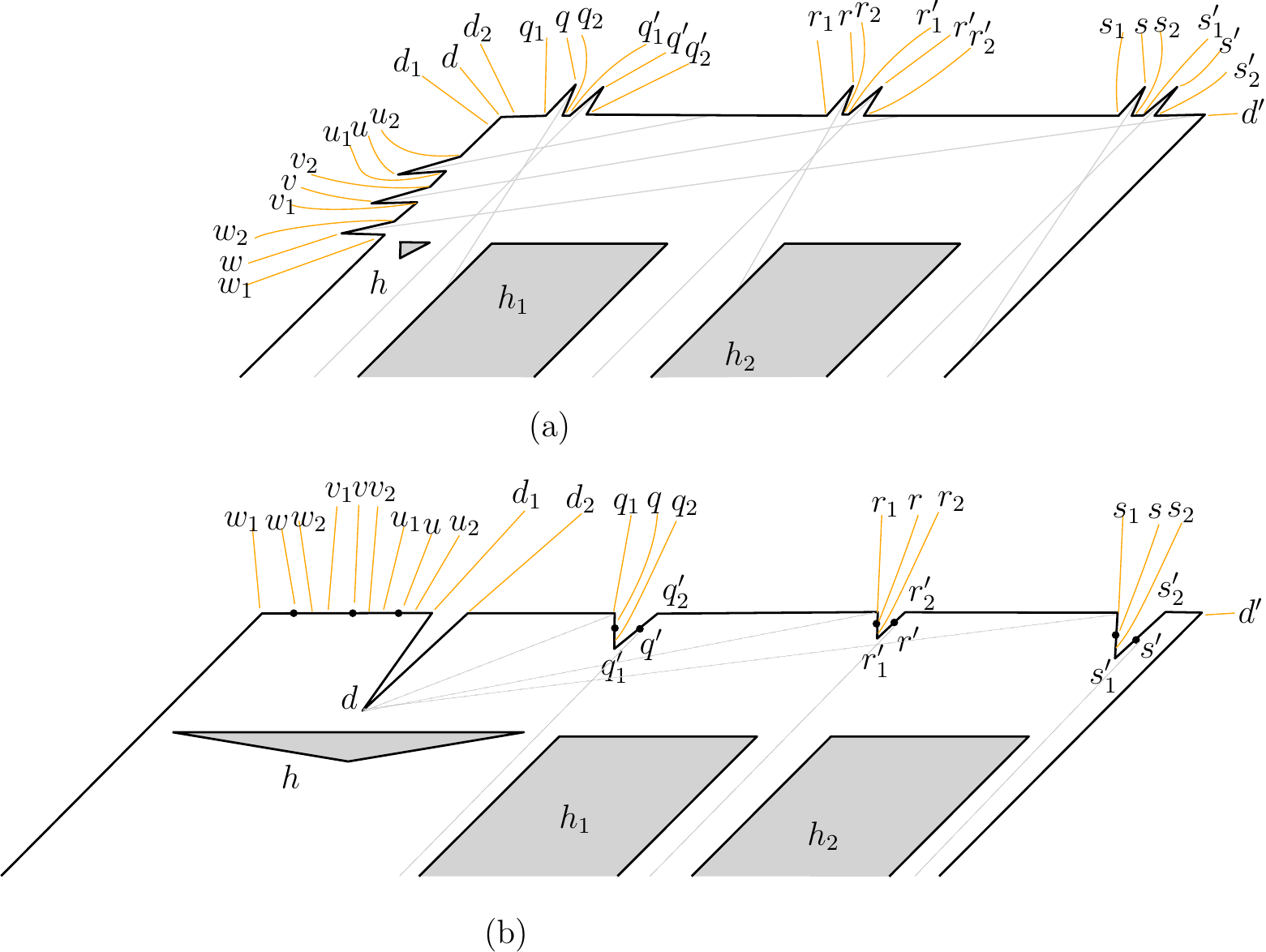}
\caption{(a)--(b) Schematic representations for  $R_1(c)$ and $R_2(c)$.}
\label{fig:details}
\end{figure}
  
Finally, we place the hole $h$ in the first channel of $R_1(c)$ 
 such that the base of $h$ is aligned 
 with the top sides of $h_1$ and $h_2$. Furthermore, we ensure that $h$ remains  
 to the left of the visibility region of $q'$.

The construction of the channels and inward dents for $R_2(c)$ is simpler compared to 
  $R_1(c)$. 
 We choose the inward dent $D$ with peak $d$ such that the 
 entire left side of each remaining inward  dent is visible to $d$,
 as illustrated using dotted lines in Figure~\ref{fig:details}(b).
 Furthermore, we ensure that exactly one $\mu$-point is visible to each of $q',r',s'$. 
Figure~\ref{fig:details}(b) illustrates these visibilities with dashed lines.

  
The placement of the triangular hole $h$ is similar to that of 
 $R_1(c)$. Here we ensure an additional constraint that the base of 
 $h$ must be large enough to block any $1$-bend visibility between 
 $\{u,v,w\}$ and $\{q',r',s'\}$.

\section{Conclusion}
We have proved that computing compatible triangulations with at most $k$ Steiner points is NP-hard for polygons with holes.  The following questions are open:

\begin{enumerate}
\item Is the problem in NP?  Is it complete for existential theory of the reals~\cite{Schaefer09}?
\item What is the complexity of the problem for a pair of simple polygons?  For a pair of rectangles with points inside?
\item How hard is it to decide if two polygonal regions, or two rectangles with points inside, have compatible triangulations with no Steiner points? For simple polygons, this can be decided in polynomial-time~\cite{DBLP:journals/comgeo/AronovSS93}.
\end{enumerate}
 


\bibliographystyle{abbrv}
\bibliography{ref}

\begin{thebibliography}{10}

\bibitem{aichholzer2003towards}
O.~Aichholzer, F.~Aurenhammer, F.~Hurtado, and H.~Krasser.
\newblock Towards compatible triangulations.
\newblock {\em Theoretical Computer Science}, 296(1):3--13, 2003.

\bibitem{alamdari2016morph}
S.~Alamdari, P.~Angelini, F.~Barrera-Cruz, T.~M. Chan, G.~Da~Lozzo,
  G.~Di~Battista, F.~Frati, P.~Haxell, A.~Lubiw, M.~Patrignani, V.~Roselli,
  S.~Singla, and B.~T. Wilkinson.
\newblock How to morph planar graph drawings.
\newblock {\em to appear in SIAM Journal on Computing}, 2017.

\bibitem{DBLP:journals/comgeo/AronovSS93}
B.~Aronov, R.~Seidel, and D.~L. Souvaine.
\newblock On compatible triangulations of simple polygons.
\newblock {\em Comput. Geom.}, 3:27--35, 1993.

\bibitem{DBLP:conf/cccg/BabikovSW97}
M.~Babikov, D.~L. Souvaine, and R.~Wenger.
\newblock Constructing piecewise linear homeomorphisms of polygons with holes.
\newblock In {\em Proceedings of the 9th Canadian Conference on Computational
  Geometry, Kingston, Ontario, Canada}, 1997.

\bibitem{baxter2009compatible}
W.~V. Baxter~III, P.~Barla, and K.-i. Anjyo.
\newblock Compatible embedding for 2{D} shape animation.
\newblock {\em IEEE Transactions on Visualization and Computer Graphics},
  15(5):867--879, 2009.

\bibitem{DBLP:journals/jgaa/ChanFGLMS15}
T.~M. Chan, F.~Frati, C.~Gutwenger, A.~Lubiw, P.~Mutzel, and M.~Schaefer.
\newblock Drawing partially embedded and simultaneously planar graphs.
\newblock {\em Journal of Graph Algorithms and Applications}, 19(2):681--706,
  2015.

\bibitem{deBerg2010}
M.~de~Berg and A.~Khosravi.
\newblock Optimal binary space partitions in the plane.
\newblock In {\em Proceedings International Computing and Combinatorics
  Conference (COCOON 2010)}, volume 6196 of {\em LNCS}, pages 216--225.
  Springer, 2010.

\bibitem{DBLP:journals/talg/GuptaW07}
H.~Gupta and R.~Wenger.
\newblock Constructing pairwise disjoint paths with few links.
\newblock {\em {ACM} Transactions on Algorithms}, 3(3):26, 2007.

\bibitem{DBLP:journals/ijcga/KranakisU99}
E.~Kranakis and J.~Urrutia.
\newblock Isomorphic triangulations with small number of {S}teiner points.
\newblock {\em International Journal of Computational Geometry \&
  Applications}, 9(2):171--180, 1999.

\bibitem{DBLP:journals/algorithmica/PachSS96}
J.~Pach, F.~Shahrokhi, and M.~Szegedy.
\newblock Applications of the crossing number.
\newblock {\em Algorithmica}, 16(1):111--117, 1996.

\bibitem{DBLP:journals/gc/PachW01}
J.~Pach and R.~Wenger.
\newblock Embedding planar graphs at fixed vertex locations.
\newblock {\em Graphs and Combinatorics}, 17(4):717--728, 2001.

\bibitem{DBLP:conf/compgeom/Saalfeld87}
A.~Saalfeld.
\newblock Joint triangulations and triangulation maps.
\newblock In {\em Proceedings of the Third Annual Symposium on Computational
  Geometry (SoCG)}, pages 195--204. {ACM}, 1987.

\bibitem{Schaefer09}
M.~Schaefer.
\newblock Complexity of some geometric and topological problems.
\newblock In {\em 17th International Symposium on Graph Drawing, (GD 2009)},
  volume 5849 of {\em LNCS}, pages 334--344. Springer, 2010.

\bibitem{sw1994}
D.~L. Souvaine and R.~Wenger.
\newblock Constructing piecewise linear homeomorphisms.
\newblock Technical report, DIMACS, New Brunswick, New Jersey, 1994.

\bibitem{surazhsky2004high}
V.~Surazhsky and C.~Gotsman.
\newblock High quality compatible triangulations.
\newblock {\em Engineering with Computers}, 20(2):147--156, 2004.

\bibitem{thomassen1983deformations}
C.~Thomassen.
\newblock Deformations of plane graphs.
\newblock {\em Journal of Combinatorial Theory, Series B}, 34(3):244--257,
  1983.

\end{thebibliography}

\end{document}